\documentclass[hidelinks,final,5p,times]{elsarticle}
\usepackage{graphics} % for pdf, bitmapped graphics files
\usepackage{epsfig} % for postscript graphics files
\usepackage{mathptmx} % assumes new font selection realization installed
\usepackage{times} % assumes new font selection realization installed
\usepackage{amsthm}
\usepackage{amssymb, amsmath, mathrsfs}
\usepackage{epstopdf} 
\usepackage{thmtools}
\usepackage{etoolbox}
\usepackage{diagbox}
\usepackage{subeqnarray}
\usepackage[citecolor=black,linkcolor=black,urlcolor=black,colorlinks=false]{hyperref}
\usepackage{multicol}
\usepackage{cases}
\biboptions{numbers,sort&compress}
\DeclareMathOperator{\rank}{rank} %
\DeclareMathOperator{\diag}{diag} %
\DeclareMathOperator{\im}{Im} %

\allowdisplaybreaks

\usepackage{calc,pifont} 
\newcounter{local}
\renewcommand\theenumi{\protect\setcounter{local}%
  {171+\the\value{enumi}}\protect\ding{\value{local}}}

\theoremstyle{definition}
\newtheorem{thm}{Theorem}
\newtheorem{definition}{Definition}
\newtheorem{lem}{Lemma}
\newtheorem{rmk}{Remark}
\newtheorem{exam}{Example}

\newtheorem{coro}{Corollary}

\DeclareMathAlphabet{\mathpzc}{OT1}{pzc}{m}{it}

\begin{document}
\begin{frontmatter}
\title{Linear quantum systems with  diagonal passive Hamiltonian and a single dissipative channel\tnoteref{footnoteinfo}}
\tnotetext[footnoteinfo]{This work was supported by the Australian Research Council and the
Australian Academy of Science.}
\author[Australia]{Shan Ma}\ead{shanma.adfa@gmail.com}
\author[Australia]{Ian R. Petersen}\ead{i.r.petersen@gmail.com}
\author[Australia]{Matthew J. Woolley}\ead{m.woolley@adfa.edu.au}
\address[Australia]{School of Engineering and Information Technology, University of New South Wales at the Australian Defence Force Academy, Canberra, ACT 2600, Australia}

\begin{abstract}
Given any covariance matrix corresponding to a so-called pure Gaussian state, a linear quantum system can be designed to achieve the assigned covariance matrix. In most cases, however, one might obtain a system that is difficult to realize in practice.  In this paper, we restrict our attention to a special class of linear quantum systems, i.e., systems with diagonal passive Hamiltonian and a single dissipative channel.  The practical implementation of such a system would be relatively simple. We then parametrize the class of pure Gaussian state covariance matrices that can be achieved by this particular type of linear quantum system. 
\end{abstract}

\begin{keyword}
Linear quantum system, Covariance assignment, Pure Gaussian state, Diagonal Hamiltonian, Passive Hamiltonian, Single dissipative channel.  
\end{keyword}
\end{frontmatter}

%%%%%%%%%%%%%%%%%%%%%%%%%%%%%%%%%%%%%%%%%%%
%%%%%%%%%%%%%%%%%%%%%%%%%%%%%%%%%%%%%%%%%%%
%%%%%%%%%%%%%%%%%%%%%%%%%%%%%%%%%%%%%%%%%%%

\section{Introduction}
 The  covariance assignment problem was first explicitly described in~\cite{HS87:ijc}; since then it has been extensively studied in the stochastic control literature, e.g.,~\cite{CS87:tac,WD90:tac,GS97:auto,WHH01:tsp}. The motivation for covariance assignment comes from the fact that many stochastic systems have performance goals that are expressed in terms of the variances (or covariance matrix) of the system states. So by assigning an appropriate matrix value to the system state covariance, desired performance goals could be achieved.  The theory of covariance assignment has applications in model reduction, system identification, and state estimation. 

An analogous idea of covariance assignment has recently been used in the context of quantum systems for the generation of pure Gaussian states~\cite{KY12:pra,Y12:ptrsa,IY13:pra,MWPY14:msc}. A pure state is a state about which we can have the maximum amount of knowledge allowed by quantum mechanics. Pure 
Gaussian states are key resources for continuous-variable quantum information processing and quantum computation~\cite{BP03:book}. A Gaussian state (with zero mean) is completely specified by its covariance matrix. Mathematically, an $N$-mode Gaussian state is pure if and only if its
covariance matrix $V$ meets $\det(V)=2^{-2N}$. The generation of a zero-mean pure Gaussian state is in fact a covariance assignment problem where the goal is to construct  a linear quantum system that is strictly stable and achieves  a steady-state covariance matrix  corresponding to a desired pure Gaussian state. Unlike the classical covariance assignment problem, which typically involves designing feedback controllers, 
pure Gaussian states may be generated via synthesis of a linear quantum system that achieves an assigned covariance matrix, without using any explicit feedback control.
  
It was shown in references~\cite{KY12:pra,Y12:ptrsa,IY13:pra, MWPY14:msc} that given any pure Gaussian state, a linear quantum system can be designed to achieve the covariance matrix corresponding to this pure Gaussian state. Furthermore,  parametrizations of the system Hamiltonian $\hat{H}$ and the dissipation $\hat{L}$  were also developed for the system; see~\cite{KY12:pra,Y12:ptrsa} for details. According to this result, for most pure Gaussian states, linear quantum systems that generate pure Gaussian states might have a complex structure that is difficult to realize in practice. The system Hamiltonian $\hat{H}$ may have too many couplings or the dissipation $\hat{L}$ contains too many dissipative processes. Thus, one would like to  characterize the class of pure Gaussian states that can be generated by a simple class of linear quantum systems. 

In this paper, we impose some constraints on a  linear quantum system to ensure that the resulting system is relatively easy to realize. The first constraint is
imposed on the system Hamiltonian $\hat{H}$; we assume that the system  Hamiltonian $\hat{H}$ is of the form $\hat{H}=\sum_{j=1}^{N}  \frac{\omega_{j}}{2}\left(\hat{q}_{j}^{2}+\hat{p}_{j}^{2}\right)$, where each $\omega_{j}$ can be an arbitrary real number and $(\hat{q}_{j},\; \hat{p}_{j})$ are the position and momentum operators for the $j$th mode. This describes a collection of $N$ independent quantum harmonic oscillators.  If we write the system Hamiltonian $\hat{H}$ as a quadratic form, i.e., $\hat{H}=\frac{1}{2}\hat{x}^{\top}G\hat{x}$,
where $\hat{x}=\left[\hat{q}_{1}\;\cdots\; \hat{q}_{N}\;\hat{p}_{1}\;\cdots\;\hat{p}_{N}\right]^{\top}$ and $G=G^{\top}\in \mathbb{R}^{2N\times 2N}$, we will find that $G$ is a diagonal matrix, describing a passive system Hamiltonian~\cite{P11:auto,P12:scl,GY16:tac,GZ15:auto}. 
Another constraint is imposed on the dissipation $\hat{L}$; we assume that the system is only coupled to a single designed dissipative environment.  Under the constraints above,  the resulting linear quantum system will be relatively easy to implement in practice. We then parametrize the class of pure Gaussian states that can be generated by this particular type of  linear quantum system. 

The remainder of the paper is organized as
follows. In Section~\ref{PRELIMINARIES}, we review some relevant properties of pure Gaussian states and recall some recent results on  covariance assignment corresponding to pure Gaussian states. In Section~\ref{Constraints}, we give the system constraints explicitly. We also explain their physical meanings. In Section~\ref{Characterization}, we parametrize the class of pure Gaussian states that can be  generated by linear quantum systems subject to the constraints described in Section~\ref{Constraints}. In Section~\ref{Example}, we provide three examples to illustrate the main results. In these examples, we also consider the impact on the system of uncontrolled couplings to additional dissipative environments.   Section~\ref{Conclusion} concludes the paper.

\textit{Notation.} For a matrix $A=[A_{jk}]$ whose entries $A_{jk}$ are complex numbers or 
operators, we define $A^{\dagger}=[A_{kj}^{\ast}]$, $A^{\top}=[A_{kj}]$, 
 where the superscript ${}^{\ast}$ denotes either 
the complex conjugate of a complex number or the adjoint of an operator. 
$\mathbb{R}^{m \times n}$ denotes the set of real $m \times n$ matrices. 
$\mathbb{C}^{m \times n}$ denotes the set of complex $m \times n$ matrices. 
$\mathbb{Z}_{+}$ denotes the set of positive integers. For a real number $x>0$, $\lfloor x \rfloor$ denotes the largest integer not greater than $x$. 
For a real symmetric matrix $A=A^{\top}$, $A> 0$  means that $A$ is positive definite.  $\diag[A_{1},\cdots,A_{n}]$ denotes a block diagonal matrix with diagonal blocks $A_{j}$, $j=1,2,\cdots,n$.

\section{Preliminaries}\label{PRELIMINARIES}
Consider a linear quantum system of $N$ modes. Each mode is characterized by a pair of quadrature operators $\{\hat{q}_{j},\; \hat{p}_{j}\}$, $j=1,2,\cdots,N$, which satisfy the following commutation relations (we use the units $\hbar=1$)
\begin{align*}
\left[\hat{q}_{j}, \hat{p}_{k}\right]=i\delta_{jk}, \quad j,k=1,2,\cdots,N.
\end{align*}
Here $\delta_{jk}$ is the Kronecker delta. If we collect all the quadrature operators of the system into an operator-valued vector $\hat{x}\triangleq\left[\hat{q}_{1}\;\cdots\; \hat{q}_{N}\;\hat{p}_{1}\;\cdots\;\hat{p}_{N}\right]^{\top}$, then the commutation relations above can be rewritten as
\begin{align*}
\hat{x}\hat{x}^{\top}-\left(\hat{x}\hat{x}^{\top}\right)^{\top}=i\Sigma_{N},\quad \Sigma_{N}\triangleq\begin{bmatrix}
0 & I_{N}\\
-I_{N} &0
\end{bmatrix}.
\end{align*}
Suppose that the system Hamiltonian $\hat{H}$ is quadratic in  $\hat{x}$, i.e., 
$\hat{H}=\frac{1}{2}\hat{x}^{\top}G\hat{x}$,
with $G=G^{\top}\in \mathbb{R}^{2N\times 2N}$, and the coupling vector $\hat{L}$ is  linear  in $\hat{x}$, i.e., 
$\hat{L}=C\hat{x}$, with $C\in \mathbb{C}^{K\times 2N}$, then the time evolution of the quantum system can be described by the following quantum stochastic differential equations (QSDEs)
\begin{equation} \label{QSDE1}
\left\{\begin{aligned}
d\hat{x}(t)&=\mathcal{A}\hat{x}(t)dt+\mathcal{B} \begin{bmatrix}
d\hat{A}^{\top}(t) &d\hat{A}^{\dagger}(t)
\end{bmatrix}^{\top},\\
d\hat{Y}(t)&=\mathcal{C}\hat{x}(t)dt+d\hat{A}(t),
 \end{aligned}\right.
\end{equation}
where $\mathcal{A}=\Sigma_{N}(G+\im(C^{\dagger}C))$, 
$\mathcal{B}=i\Sigma_{N}[- C^{\dagger}\; \; C^{\top}]$, and $\mathcal{C}=C$~
\cite{Y12:ptrsa,N10:tac,JNP08:tac,NJP09:auto,NJD09:siamjco}. The input $\hat{A}(t)$ consists of $K$ independent quantum stochastic processes, i.e., $\hat{A}(t) =\left[\hat{A}_{1}(t)\; \hat{A}_{2}(t)\; \cdots \; \hat{A}_{K}(t)\right]^{\top}$ 
with $\hat{A}_{j}(t)$, $j=1,2,\cdots,K$, satisfying the following quantum It\=o rules: 
\begin{equation} \label{ito}
\left\{\begin{aligned} 
d\hat{A}_{j}(t)d\hat{A}_{k}^{\ast}(t)&=\delta_{jk}dt,\; d\hat{A}_{j}(t)d\hat{A}_{k}(t)=0,\\
d\hat{A}_{j}^{\ast}(t)d\hat{A}_{k}^{\ast}(t)&=d\hat{A}_{j}^{\ast}(t)d\hat{A}_{k}(t)=0.
 \end{aligned}\right. 
\end{equation}
The output $\hat{Y}(t)=\left[\hat{Y}_{1}(t)\;\; \hat{Y}_{2}(t)\;\;\cdots \;\; \hat{Y}_{K}(t)\right]^{\top}$ 
satisfies quantum It\=o rules similar to~\eqref{ito}~\cite{N10:tac,NJD09:siamjco,JNP08:tac,GJ09:tac, NJP09:auto,Y12:ptrsa,TN15:scl}. 
The quantum expectation value of the vector $\hat{x}$ is denoted by $
\langle \hat{x} \rangle$, 
and the covariance matrix is denoted by 
$V=\frac{1}{2}\langle \triangle\hat{x}{\triangle\hat{x}}^{\top}+(\triangle\hat{x}{\triangle\hat{x}}^{\top})^{\top} \rangle$, where $ \triangle\hat{x}=\hat{x}-\langle \hat{x} \rangle$~\cite{KY12:pra,OHY11:cdc,Y12:ptrsa,IY13:pra, MWPY14:msc,MFL11:pra,TN15:scl}. 
 By using the quantum It\=o rule, it can be shown that the mean value $\langle \hat{x} \rangle$ and the covariance matrix $V$ obey the following dynamical equations
\begin{equation} \label{QSDEequation}
\left\{\begin{aligned}
\frac{d\langle\hat{x}(t)\rangle}{dt}&=\mathcal{A}\langle\hat{x}(t)\rangle,  \\
\frac{dV(t)}{dt}&=\mathcal{A}V(t)+V(t)\mathcal{A}^{\top}+\frac{1}{2}\mathcal{B}\mathcal{B}^{\dagger}.   
\end{aligned}\right.
\end{equation}
 A Gaussian state is completely specified by $\langle \hat{x} \rangle$ and $V$.   As the mean value $\langle \hat{x}\rangle$ contains no information about noise and entanglement, we will restrict our attention to zero-mean Gaussian states. The purity of a Gaussian state is given by $p\triangleq\frac{1}{2^{N}\sqrt{\det(V)}}$. A Gaussian state is pure if and only if $\det(V)=2^{-2N}$~\cite{KY12:pra}. In fact, if a Gaussian state is pure, its covariance matrix $V$ can always be factored as 
\begin{align}\label{covariance}
V=\frac{1}{2}\begin{bmatrix}
Y^{-1} &Y^{-1}X\\
XY^{-1} &XY^{-1}X+Y
\end{bmatrix},
\end{align}
where $X=X^{\top}\in \mathbb{R}^{N\times N}$ and $Y=Y^{\top}\in \mathbb{R}^{N\times N}$ and $Y>0$~\cite{MFL11:pra,KY12:pra,Y12:ptrsa}.   As can be seen from~\eqref{covariance}, a  pure Gaussian state (with zero mean) is uniquely specified by a complex, symmetric matrix $Z\triangleq X+iY$, where $X=X^{\top}\in \mathbb{R}^{N\times N}$, $Y=Y^{\top}\in \mathbb{R}^{N\times N}$ and $Y>0$. In other words, given a zero-mean  pure Gaussian state, the matrix $Z=X+iY$ can be uniquely determined, and vice versa. In the following discussions, we will refer to $Z$ as the Gaussian graph matrix for  a pure Gaussian state~\cite{MFL11:pra}. If the system~\eqref{QSDE1} is initially in a Gaussian state, then the system will forever remain in a Gaussian state, with the first two moments $\langle \hat{x} \rangle$ and $V$ evolving as described in~\eqref{QSDEequation}. In order to generate a pure Gaussian state, $\mathcal{A}$ must be  Hurwitz, i.e., every eigenvalue of $\mathcal{A}$ has a negative real part. Then using~\eqref{QSDEequation}, we have $\langle\hat{x}(\infty)\rangle =0$ and $V(\infty)$ satisfies the following Lyapunov equation 
\begin{align}
\mathcal{A}V(\infty)+V(\infty)\mathcal{A}^{\top}+\frac{1}{2}\mathcal{B}\mathcal{B}^{\dagger}=0. \label{lya}
\end{align}
If the value of $V(\infty)$ obtained from~\eqref{lya} is identical to~\eqref{covariance}, then we can conclude that the desired pure Gaussian state is  generated by the system as its steady state. From the discussions above, the pure Gaussian state generation problem is indeed a covariance assignment problem, where the goal is to construct a system Hamiltonian $\hat{H}$ and a coupling vector $\hat{L}$, such that the system described by~\eqref{QSDE1} is strictly stable and  the covariance matrix $V$ corresponding to the desired pure Gaussian  state is the unique solution of the Lyapunov equation~\eqref{lya}. Recently, a necessary and sufficient condition has been developed in~\cite{KY12:pra,Y12:ptrsa} for solving this problem. The result is summarized as follows. 
\begin{lem}[\cite{KY12:pra,Y12:ptrsa}]\label{lem1}
Let $V$ be the covariance matrix corresponding to an $N$-mode pure Gaussian state.  Assume that it is expressed in the factored form~\eqref{covariance}. Then this pure Gaussian state is  generated by the linear quantum system~\eqref{QSDE1} if and only if
\begin{align} \label{G}
G=\begin{bmatrix}
XRX+YRY-\Gamma Y^{-1}X-XY^{-1}\Gamma^{\top} &-XR+\Gamma Y^{-1}\\
-RX+Y^{-1}\Gamma^{\top} &R
\end{bmatrix},
\end{align}
 and 
\begin{align} \label{C}
C=P^{\top}\left[-(X+iY)\;\; I_{N}\right], 
\end{align} 
where  $R=R^{\top}\in\mathbb{R}^{ N\times N}$,  $\Gamma=-\Gamma^{\top}\in\mathbb{R}^{ N\times N}$, and  $P\in \mathbb{C}^{N\times K}$ are free matrices satisfying the following rank condition
\begin{align}
\rank\left([P\;\;\;QP\;\;\;\cdots\;\;\;Q^{N-1}P]\right)=N,\;\;\;  Q\triangleq -iRY+Y^{-1}\Gamma. \label{rankconstraint}
\end{align}
\end{lem}

\begin{rmk}
From Lemma~\ref{lem1}, we could conjecture that for most pure Gaussian states,  the corresponding linear quantum systems that generate the pure Gaussian states might not be easy to realize in practice. Either the system Hamiltonian $\hat{H}=\frac{1}{2}\hat{x}^{\top}G\hat{x}$ or the coupling vector $\hat{L}=C\hat{x}$ might have a complex structure. We present an example to illustrate this fact. 
\end{rmk}

\begin{exam}
{\it Gaussian cluster states} are an important class of pure Gaussian states~\cite{BP03:book,MFL11:pra}. 
The covariance matrix  $V$ corresponding to an $N$-mode Gaussian cluster state is given by
$
V=\frac{1}{2}\begin{bmatrix}
e^{2\alpha}I_{N} &e^{2\alpha}B\\
e^{2\alpha}B &e^{-2\alpha}I_{N}+e^{2\alpha}B^{2}
\end{bmatrix}$, 
where $B=B^{\top}\in\mathbb{R}^{ N\times N}$ and $\alpha\in\mathbb{R}$ is 
the squeezing parameter.   
Applying the factorization~\eqref{covariance} yields $X=B$ and $Y=e^{-2\alpha}I_{N}$. Let us consider a simple case where 
\begin{align}
\label{first cluster example}
X=\begin{bmatrix}
0 &1 &0 \\
1 &0 &1  \\
0 &1 &0 
\end{bmatrix}, \quad Y=e^{-2\alpha}I_{3}. 
\end{align}
Using Lemma~\ref{lem1}, we can construct a linear quantum system that  generates the Gaussian cluster state~\eqref{first cluster example}. For example, let us choose $R=I_{3}$, $\Gamma=e^{-2\alpha}\begin{bmatrix}
0 &1 &0\\
-1 &0 &1\\
0 &-1 &0
\end{bmatrix}$ and $P=\begin{bmatrix}
1 &0 \\
0 &0 \\
0 &1 
\end{bmatrix}$. Then, by direct substitution, we can verify that the rank condition~\eqref{rankconstraint} holds. In this case, using Lemma~\ref{lem1}, the system Hamiltonian  is given by
\begin{align*}
\hat{H}&=\frac{1}{2}(-1+e^{-4\alpha})\hat{q}_{1}^{2}+\frac{1}{2}(2+e^{-4\alpha})\hat{q}_{2}^{2}+\frac{1}{2}(3+e^{-4\alpha})\hat{q}_{3}^{2}\\&\;+\frac{1}{2}\hat{p}_{1}^{2}+\frac{1}{2}\hat{p}_{2}^{2}+\frac{1}{2}\hat{p}_{3}^{2}+\hat{q}_{1}\hat{q}_{3}-2\hat{q}_{2}\hat{p}_{1}-2 \hat{q}_{3}\hat{p}_{2},
\end{align*}
and the coupling vector is given by $\hat{L}=\begin{bmatrix}
-e^{-2\alpha}i\hat{q}_{1}-\hat{q}_{2}+\hat{p}_{1}\\
-e^{-2\alpha}i\hat{q}_{3}-\hat{q}_{2}+\hat{p}_{3}
\end{bmatrix}$. We see that, in the Hamiltonian $\hat{H}$, each mode is coupled to the other two modes. Moreover,  
 $\hat{L}$ contains two dissipative processes. These features indicate that the implementation of this system could be rather challenging in practice. 
\end{exam}

\section{Constraints} \label{Constraints}
In this section, we put some restrictions on a  linear quantum system to ensure that the resulting system is relatively easy to implement in practice. We assume that the system~\eqref{QSDE1} is subject to the following two constraints:
\begin{enumerate}
\item The Hamiltonian $\hat{H}$ is of the form 
$\hat{H}=\sum_{j=1}^{N}  \frac{\omega_{j}}{2}(\hat{q}_{j}^{2}+\hat{p}_{j}^{2})$, where each $\omega_{j}$ can be an arbitrary real number. \label{constraint1}
\item The system is coupled to a single dissipative environment, i.e., $K=1$. \label{constraint2}
\end{enumerate}

The system Hamiltonian  can be written as $\hat{H}=\sum\limits_{j=1}^{N}  \frac{\omega_{j}}{2}(\hat{q}_{j}^{2}+\hat{p}_{j}^{2})=\sum\limits_{j=1}^{N}  \frac{\omega_{j}}{2}(\hat{a}_{j}\hat{a}_{j}^{\ast}+\hat{a}_{j}^{\ast}\hat{a}_{j})=\sum\limits_{j=1}^{N}  \frac{\omega_{j}}{2}(2\hat{a}_{j}^{\ast}\hat{a}_{j}+1)\cong\sum\limits_{j=1}^{N}\omega_{j} \hat{a}_{j}^{\ast}\hat{a}_{j}$, where $\hat{a}_{j}=(\hat{q}_{j}+i\hat{p}_{j})/\sqrt{2}$ and $\hat{a}_{j}^{\ast}=(\hat{q}_{j}-i\hat{p}_{j})/\sqrt{2}$ denote the annihilation operator and creation operator for the $j$th mode, respectively.  Note that here we have used the fact that a constant term in the Hamiltonian does not affect the dynamics of a quantum system, and hence can be dropped. The Hamiltonian describes a collection of $N$ independent quantum harmonic oscillators. As can be seen in~\ref{constraint1}, the system Hamiltonian $\hat{H}$ is a sum of $N$ independent harmonic oscillator terms and no couplings exist between these oscillators. 

The second constraint~\ref{constraint2} requires us to use only one dissipative process to generate a pure Gaussian state.  Note that when $K = 0$ (that is, no dissipation is introduced), the system~\eqref{QSDE1} reduces to an isolated quantum system. Any isolated quantum system cannot be strictly
stable and hence cannot evolve into a
pure Gaussian steady state. In order for a quantum system subject to~\ref{constraint1} and~\ref{constraint2} to be strictly stable, the single reservoir must act globally on all system modes, since otherwise the system contains an isolated subsystem which cannot be strictly stable.

\section{Characterization} \label{Characterization}
In this section, we characterize the class of pure Gaussian states that can be  generated  by  linear open quantum systems subject to the two constraints~\ref{constraint1} and~\ref{constraint2}. First, we introduce some  technical results that will be used to derive the main results.

\begin{lem}\label{diagonalizable}
Let $A\in\mathbb{C}^{n \times n}$, and assume that $A^{2}=\epsilon I_{n}$, $\epsilon > 0$. Then $A$ is diagonalizable over $\mathbb{C}$ and the eigenvalues of  $A$ are in the set $\{\pm\sqrt{\epsilon}\}$.
\end{lem}
\begin{proof}
First we show that the eigenvalues of $A$ are either $\sqrt{\epsilon}$ or $-\sqrt{\epsilon}$. Suppose that $\lambda$ is an eigenvalue of $A$ and that $A\xi = \lambda \xi$, $\xi \in 
\mathbb{C}^{n }$, $\xi \ne 0$. Then we have $A^{2}\xi = \lambda A \xi=\lambda^{2} \xi =\epsilon \xi$. So $\left(\lambda^{2}-\epsilon\right) \xi = 0$. Since $\xi \ne 0$, it follows that $\lambda=\pm\sqrt{\epsilon}$.

Next we show that $A$ is diagonalizable. Since $A^{2}=\epsilon I_{n}$, we have $A^{2}-\epsilon I_{n}=0 $ and hence $p(s)=(s+\sqrt{\epsilon})(s-\sqrt{\epsilon})$ is a monic polynomial of degree $2$ that annihilates $A$. The minimal polynomial $q_{A}(s)$ of $A$ divides $p(s)$, and hence $q_{A}(s)=s+\sqrt{\epsilon},\;s-\sqrt{\epsilon},\;\text{or}\;(s+\sqrt{\epsilon})(s-\sqrt{\epsilon})$. For all cases, $q_{A}(s)$ is a product of linear factors, with no repetitions. As a result, the Jordan canonical form of $A$ consists only of Jordan blocks of size one~\cite[Theorem 3.3.6]{HJ12:book}. That is, $A$ is diagonalizable.  
\end{proof}

\begin{rmk}
An $n \times n$ matrix $A$ is said to be involutory if $A^{2}= I_{n}$~\cite[Definition 0.9.13]{HJ12:book}. 
\end{rmk}

\begin{lem} \label{identity}
Let $A=A_{1}+iA_{2}$ with $A_{1}\in \mathbb{R}^{n\times n}$, $ A_{2} \in \mathbb{R}^{n\times n}$ and $A_{2}=A_{2}^{\top}>0$. Suppose  $A^{2}=-I_{n}$. Then $A=iI_{n}$. 
\end{lem}

\begin{proof}
Substituting $A=A_{1}+iA_{2}$ into $A^{2}=-I_{n}$, we obtain $
A_{1}^{2}-A_{2}^{2}=-I_{n}$ and $A_{1}A_{2}+A_{2}A_{1}=0$.  Since $A_{2}>0$, it follows from Sylvester's theorem that $A_{1}=0$~\cite[Theorem 2.4.4.1]{HJ12:book}. Hence $A_{2}^{2}=I_{n}$. According to Lemma~\ref{diagonalizable}, $A_{2}$ is diagonalizable and its eigenvalues are either $1$ or $-1$. But $A_{2}>0$, so $A_{2}$ can only be an identity matrix, i.e., $A_{2}=I_{n}$. Therefore $A=iI_{n}$.
\end{proof}

\begin{lem} \label{twodimensional}
Let $A=A_{1}+iA_{2}$ with $A_{1}\in \mathbb{R}^{2\times 2}$, $A_{2} \in \mathbb{R}^{2\times 2}$ and $A_{2}=A_{2}^{\top}>0$.  Suppose  $ \left(\diag[1,-1]A\right)^{2}=-I_{2}$. Then the two eigenvalues of the matrix $ \diag[1,-1]A$ are $\lambda_{1}=i$ and $\lambda_{2}=-i$.
\end{lem}

\begin{proof}
Since $\left(\diag[1,-1]A\right)^{2}=-I_{2}$, by Lemma~\ref{diagonalizable}, the matrix $\diag[1,-1]A$ is diagonalizable and its eigenvalues are either $i$ or $-i$. If all the eigenvalues of $\diag[1,-1]A$ are $i$, then we have $\diag[1,-1]A=iI_{2}$, and hence $A=i\diag[1,-1]$. In this case, we have $A_{2}=\diag[1,-1]$, which contradicts the assumption that $A_{2}>0$. Similarly, if all the eigenvalues of $\diag[1,-1]A$ are $-i$, we will also obtain a contradiction. Therefore, the two eigenvalues of $ \diag[1,-1]A$ are different; they are $\lambda_{1}=i$ and $\lambda_{2}=-i$.
\end{proof}

\begin{definition} \label{def}
A matrix $A\in\mathbb{C}^{n \times n}$ is said to be non-derogatory if for any eigenvalue $\lambda$ of $A$, $\rank(A - \lambda I_{n}) = n-1$. 
\end{definition}

\begin{lem} \label{similar}
Suppose $A_{1}\in\mathbb{C}^{n \times n}$, $A_{2}\in\mathbb{C}^{n \times n}$ and  $A_{1}$ is similar to $A_{2}$. Then $A_{2}$ is a  non-derogatory matrix if and only if $A_{1}$ is a  non-derogatory matrix.
\end{lem}

\begin{proof}
Assume that $A_{1}$ is a non-derogatory matrix and $A_{1}=FA_{2}F^{-1}$, where $F$ is a non-singular matrix. We now show that $A_{2}$ is also a  non-derogatory matrix. Suppose $\lambda$ is an eigenvalue of $A_{2}$. Then $\lambda $ is also an eigenvalue of $A_{1}$. Furthermore, $\rank(A_{2} - \lambda I_{n}) = \rank(F^{-1}A_{1}F - \lambda I_{n}) =\rank(A_{1} - \lambda I_{n}) = n-1$. Hence by definition, $A_{2}$ is a  non-derogatory matrix.
\end{proof}

The following lemma can be found in~\cite{H77:jrnbs}. To make this paper self-contained, we  include its proof. 
\begin{lem} \label{non-dero}
Let $A\in\mathbb{C}^{n \times n}$. Then $A$ is a non-derogatory matrix if and only if there exists a column vector $\xi\in \mathbb{C}^{n}$ such that the pair $(A,\; \xi)$ is controllable, i.e.,
\begin{align}
\rank\left([\xi\;\;\;A\xi\;\;\;\cdots\;\;\;A^{n-1}\xi]\right)=n.  \label{controllable}
\end{align}
\end{lem}

\begin{proof}
To establish the necessity, we notice that the minimal polynomial of $A$ has degree $n$ if $A$ is a non-derogatory matrix~\cite[Theorem 3.3.15]{HJ12:book}. Assume that the minimal polynomial of $A$ is $q_{A}(s)=s^{n}+a_{n-1}s^{n-1}+a_{n-2}s^{n-2}+\cdots+a_{1}s+a_{0}$, with $a_{j}\in \mathbb{C}$, $j=0,1,\cdots, n-1$. Then it can be shown that $A$ is similar to the companion matrix of $q_{A}(s)$, i.e., 
\begin{align*}
A=F\begin{bmatrix}
0 & 0 & \dots & 0 & -a_0 \\
1 & 0 & \dots & 0 & -a_1 \\
0 & 1 & \dots & 0 & -a_2 \\
\vdots & \vdots & \ddots & \vdots & \vdots \\
0 & 0 & \dots & 1 & -a_{n-1}
\end{bmatrix}F^{-1},
\end{align*}
where $F\in\mathbb{C}^{n \times n}$ is a non-singular matrix \cite[Theorem 3.3.15]{HJ12:book}. Let $\xi=F\left[1 \;\; 0 \;\; \cdots\;\; 0\right]^{\top}$. Then we can establish that the rank constraint~\eqref{controllable} holds by direct substitution.

To establish the sufficiency, we notice that the rank constraint~\eqref{controllable} implies that $\rank\left([A-\lambda I_{n}\;\; \xi]\right)=n$ for any eigenvalue $\lambda $ of $A$~\cite[Theorem 3.1]{ZJK96:book}. Since $\xi$ is a column vector, it follows that $\rank \left(A-\lambda I_{n}\right)\ge n-1$. But if $\lambda$ is an eigenvalue of $A$, we must have $\rank \left(A-\lambda I_{n}\right)\le n-1$. Hence $\rank(A - \lambda I_{n}) = n-1$ for any eigenvalue $\lambda$ of $A$; by definition, $A$ is a non-derogatory matrix.  
\end{proof}

\begin{rmk}
According to Lemma~\ref{lem1}, the second constraint~\ref{constraint2} is equivalent to requiring that the matrix $P$ in~\eqref{C} is a column vector. According to Lemma~\ref{non-dero}, the second constraint~\ref{constraint2} is further equivalent to requiring that the matrix $Q$ in~\eqref{rankconstraint} is a non-derogatory matrix. 
\end{rmk}

Now we are in a position to present the main result of this paper. The following theorem  characterizes the class of pure Gaussian states that can be  generated by linear quantum systems subject to the two constraints~\ref{constraint1} and~\ref{constraint2}. Before presenting it, we define three sets: 
\begin{align*}
\Lambda &\triangleq\{z\;\;\big|\;\; z\in \mathbb{C}\;\;\text{and}\;\; \im(z)>0\},\\
\Pi&\triangleq\Big\{\diag[z,\;i]\;\;\big|\;\; z\in \mathbb{C} \;\; \text{and}\;\;\im(z)>0 \Big\},\\
\Xi&\triangleq\bigg\{\mathpzc{Z}\;\;\Big|\;\; \mathpzc{Z}=\mathpzc{Z}^{\top}\in \mathbb{C}^{2\times 2}, \;\;\im\left(\mathpzc{Z}\right)>0,\\
&\quad\;\;\text{and}\;\; \bigg(\diag[1,-1]
\mathpzc{Z}\bigg)^{2} =-I_{2}\bigg\}.
\end{align*}

\begin{thm} \label{thm1}
An $N$-mode pure Gaussian state can be  generated by a  linear quantum system subject to the two constraints~\ref{constraint1} and~\ref{constraint2} if and only if its Gaussian graph matrix $Z$ can be written as 
\begin{align}
Z=\mathcal{P}^{\top}\tilde{Z}\mathcal{P},\quad \tilde{Z}=\diag[\tilde{Z}_{1},\cdots,\tilde{Z}_{\lfloor \frac{N+1}{2} \rfloor}], \label{thm1Z}
\end{align}
 where $\mathcal{P}\in\mathbb{R}^{N\times N}$ is a permutation matrix, $\tilde{Z}_{1}\in\left(\Lambda\cup\Pi\cup\Xi\right)$, and $\tilde{Z}_{j}\in \Xi$, $j=2,\cdots,\lfloor \frac{N+1}{2} \rfloor$.
\end{thm}
\begin{proof}
To establish the necessity, we assume that a pure Gaussian state can be  generated by a  linear quantum system subject to the two constraints~\ref{constraint1} and~\ref{constraint2}. We will show that the Gaussian graph matrix $Z$ of this pure Gaussian state can be written in the form of the equation~\eqref{thm1Z}. First, consider the system constraint~\ref{constraint1}, which is equivalent to saying that the Hamiltonian matrix is
\begin{align*}
G=\diag[\omega_{1},\; \omega_{2},\; \cdots, \omega_{N},\;\omega_{1},\; \omega_{2},\; \cdots, \omega_{N}]. %\label{G2}
\end{align*} 
 Using Lemma~\ref{lem1}, we have
\begin{numcases}{}
R=\diag[\omega_{1},\; \omega_{2},\; \cdots, \omega_{N}], \label{first constraint1}\\
XRX +YRY-\Gamma Y^{-1}X-XY^{-1}\Gamma^{\top}  = R,\label{first constraint2}\\
-X  R +\Gamma Y^{-1}=0. \label{first constraint3}
\end{numcases}
From~\eqref{first constraint3}, we obtain $\Gamma=XRY$. 
Substituting this into~\eqref{first constraint2} yields \begin{align}
YRY-XRX=R. \label{thm1-1}
\end{align}
Recall from Lemma~\ref{lem1} that $\Gamma$ is a skew symmetric matrix. Hence we have
\begin{align}
XRY+YRX=0. \label{thm1-2}
\end{align}
Combining~\eqref{thm1-1} and~\eqref{thm1-2} gives 
\begin{align}
-ZRZ=R, \label{ZRZ}
\end{align}
where $Z=X+iY$ is the Gaussian graph matrix for the pure Gaussian state. Then it follows that\footnote[1]{The authors acknowledge helpful discussions on mathoverflow.net, see \href{http://mathoverflow.net/questions/200501/quadratic-matrix-equation}{http://mathoverflow.net/questions/200501/quadratic-matrix-equation}}
\begin{align}
ZR^{2}=-ZR(ZRZ)=R^{2}Z. \label{RZ}
\end{align}
Next we use a permutation similarity to rearrange the main diagonal entries of $R$ in ascending order of their absolute values.  Suppose $\mathcal{P}_{1}\in \mathbb{R}^{N\times N }$ is a  permutation matrix such that
\begin{align*}
\bar{R}&\triangleq \mathcal{P}_{1}R\mathcal{P}_{1}^{\top}=\diag[\underbrace{\bar{\omega}_{1},\cdots, \bar{\omega}_{1}, -\bar{\omega}_{1},\cdots, -\bar{\omega}_{1}}_{p_{1}},\cdots,\\
&\hspace{3cm}  \underbrace{ \bar{\omega}_{r},\cdots, \bar{\omega}_{r}, -\bar{\omega}_{r},\cdots, -\bar{\omega}_{r}}_{p_{r}}],
\end{align*}
where $0\le\bar{\omega}_{1}< \cdots <\bar{\omega}_{r}$. Here $p_{j}\in \mathbb{Z}_{+}$ denotes the total number of $\bar{\omega}_{j}$ and $-\bar{\omega}_{j}$ that appear in the main diagonal of $R$. For example, $p_{j}=1$ means that one of $\bar{\omega}_{j}$ and $-\bar{\omega}_{j}$ does not appear in the main diagonal of $R$ and the other only appears once. Note that $\sum_{j=1}^{r}p_{j}=N$. 
The equation~\eqref{RZ} is transformed into $\bar{Z}\bar{R}^{2}=\bar{R}^{2}\bar{Z}$, where $\bar{Z}\triangleq \mathcal{P}_{1} Z\mathcal{P}_{1} ^{\top}$.  It follows immediately  that 
 $\bar{Z}$ is a block diagonal matrix, i.e., $
\bar{Z}=\diag[\bar{Z}_{1},\cdots,\bar{Z}_{r}]$,
with $\bar{Z}_{j}\in \mathbb{C}^{p_{j}\times p_{j}} $, $1\le j \le r $. From~\eqref{ZRZ}, we have $-\bar{Z}\bar{R} \bar{Z}= \bar{R}$, i.e., 
\begin{align}
&-\bar{Z}_{j}\bar{R}_{j}\bar{Z}_{j} =\bar{R}_{j}, \quad 1\le j \le r,  \label{ZRZ2}
\end{align}
where $\bar{R}_{j}\triangleq \diag[\underbrace{ \bar{\omega}_{j},\cdots, \bar{\omega}_{j}, -\bar{\omega}_{j},\cdots, -\bar{\omega}_{j}}_{p_{j}}]$ is a diagonal block in $\bar{R}$.
Then from~\eqref{ZRZ2}, we have   $\left(i\bar{R}_{j} \bar{Z}_{j} \right)^{2}=\bar{R}_{j}^{2}= \bar{\omega}_{j}^{2}I_{p_{j}}$, $1\le j \le r$.  If $\bar{\omega}_{j}=0$, since $\bar{R}_{j}=0_{p_{j}\times p_{j}}$ we have $\bar{R}_{j} \bar{Z}_{j}=0_{p_{j}\times p_{j}}$, which is a trivial diagonal matrix. If $\bar{\omega}_{j}\ne 0$, 
according to Lemma~\ref{diagonalizable}, $\bar{R}_{j} \bar{Z}_{j}$ is diagonalizable and its eigenvalues are either $i\bar{\omega}_{j}$ or $-i\bar{\omega}_{j}$.

Let us now turn to the constraint~\ref{constraint2}, which implies that the matrix $Q$ in~\eqref{rankconstraint} is a non-derogatory matrix. Since $Q=-iRY+Y^{-1}\Gamma=-iRY-RX=-RZ=-\mathcal{P}_{1}^{\top}\bar{R}\bar{Z}\mathcal{P}_{1}$, by Lemma~\ref{similar}, $\bar{R}\bar{Z}$ is a non-derogatory matrix. Since $\bar{R}\bar{Z}=\diag[\bar{R}_{1} \bar{Z}_{1}, \; \cdots, \; \bar{R}_{r} \bar{Z}_{r}]$, it follows that each diagonal block, $\bar{R}_{j} \bar{Z}_{j}$, $1 \le j\le r $, must be a non-derogatory matrix. But $\bar{R}_{j} \bar{Z}_{j}$ is diagonalizable and its eigenvalues are either $i\bar{\omega}_{j}$ or $-i\bar{\omega}_{j}$, so the size  $p_{j}$ of $\bar{R}_{j} \bar{Z}_{j}$ must be  $p_{j}=1\; \text{or}\;\; 2$. 
 
Consider the first block, $\bar{R}_{1} \bar{Z}_{1}$. If its size $p_{1}=1$, the equation~\eqref{ZRZ2} reduces to $-\bar{Z}_{1}\bar{\omega}_{1}\bar{Z}_{1} =\bar{\omega}_{1}$.   By solving it, we find 
 \begin{equation} \notag
  \bar{Z}_{1}= \left\{ 
 \begin{aligned}
 &x_{1}+iy_{1}\;\text{with}\; x_{1}\in \mathbb{R},\; y_{1}\in \mathbb{R}\;\text{and} \;y_{1}>0,\; \;\; \text{if}\;  \bar{\omega}_{1}=0,\\
 &i, \hspace{5.7cm}  \text{if}\;  \bar{\omega}_{1}\ne 0.
 \end{aligned}\right.
   \end{equation}
The two solutions above can be combined as $\bar{Z}_{1}= x_{1}+iy_{1}$ with $x_{1}\in \mathbb{R}$, $y_{1}\in \mathbb{R}$ and $y_{1}>0$.    
 If $p_{1}=2$, we shall distinguish three cases for the equation~\eqref{ZRZ2}:  
  \begin{align}
-\bar{Z}_{1}\bar{R}_{1}\bar{Z}_{1} =\bar{R}_{1}, \label{Z1equ}
 \end{align}
where  $\bar{R}_{1}=\diag[\bar{\omega}_{1},\;\bar{\omega}_{1}],\; \diag[\bar{\omega}_{1},\;-\bar{\omega}_{1}],\;\text{or}\;\; \diag[-\bar{\omega}_{1},\;-\bar{\omega}_{1}]$.  For each case, we have $\bar{\omega}_{1}\ne 0$, because otherwise it would follow that $\bar{R}_{1} \bar{Z}_{1}=0_{2\times 2}$, which is not a non-derogatory matrix. If  $\bar{R}_{1}=\diag[\bar{\omega}_{1},\;\bar{\omega}_{1}]\;  \text{or}\; \diag[-\bar{\omega}_{1},\;-\bar{\omega}_{1}]$, then the equation \eqref{Z1equ} implies  $-\bar{Z}_{1}^{2}=I_{2}$. By Lemma~\ref{identity}, $\bar{Z}_{1}=iI_{2}$ and  hence $\bar{R}_{1} \bar{Z}_{1}=\pm i\bar{\omega}_{1} I_{2}$, which is not a non-derogatory matrix. Therefore, we can only have $\bar{R}_{1}=\diag[\bar{\omega}_{1},\;-\bar{\omega}_{1}]$. Then the equation \eqref{Z1equ} reduces to $-\bar{Z}_{1}\diag[1,\;-1]\bar{Z}_{1} =\diag[1,\;-1]$. That is, $\left(\diag[1,\;-1]\bar{Z}_{1}\right)^{2}=-I_{2}$. By Lemma~\ref{twodimensional}, the two eigenvalues of $\diag[\bar{\omega}_{1},\;-\bar{\omega}_{1}]\bar{Z}_{1}$ are $i\bar{\omega}_{1}$ and $-i\bar{\omega}_{1}$, and hence it  is a non-derogatory matrix.    
So from the discussions above, we conclude that $\bar{Z}_{1}$ is either a complex number $x_{1}+iy_{1}$ with $x_{1}\in \mathbb{R}$, $y_{1}\in \mathbb{R}$ and $y_{1}>0$,  or a $2\times 2$ matrix satisfying $\left(\diag[1,\;-1]\bar{Z}_{1}\right)^{2}=-I_{2}$.  In a similar way, we can show that the remaining diagonal blocks $\bar{Z}_{j}$, $2\le j \le r$,  are either a complex number $i$ (since $\bar{\omega}_{j}\ne 0$) or a $2\times 2$ matrix satisfying $\left(\diag[1,\;-1]\bar{Z}_{j}\right)^{2}=-I_{2}$. 
 
Now we obtain that $\bar{Z}=\diag[\bar{Z}_{1},\cdots,\bar{Z}_{r}]$, where $\bar{Z}_{1}\in \Lambda \; \text{or}\; \Xi $, and $\bar{Z}_{j}=i$ or $\bar{Z}_{j}\in \Xi$, $2\le j \le r$. If the size $N$ of  $\bar{Z}$ is an even number, we can always use a permutation similarity to transform $\bar{Z}$ into $\tilde{Z}=\mathcal{P}_{2}\bar{Z}\mathcal{P}_{2}^{\top}=\diag[\tilde{Z}_{1},\cdots,\tilde{Z}_{\frac{N}{2}}]$, where $\mathcal{P}_{2}\in\mathbb{R}^{N\times N}$ is a permutation matrix, $\tilde{Z}_{1}\in \Pi\; \text{or}\; \Xi$, and 
$\tilde{Z}_{j}\in \Xi$, $ 2\le j\le\frac{N}{2}$. Here we have used the fact that $iI_{2}\in \Xi$. Similarly, if $N$ is an odd number, we can always use a permutation similarity to transform $\bar{Z}$ into  $\tilde{Z}=\mathcal{P}_{2}\bar{Z}\mathcal{P}_{2}^{\top}=\diag[\tilde{Z}_{1},\cdots,\tilde{Z}_{\frac{N+1}{2}}]$, where $\tilde{Z}_{1}\in \Lambda$, and 
$\tilde{Z}_{j}\in \Xi$, $2\le j \le \frac{N+1}{2}$.  Therefore, the matrix  $\bar{Z}$ is permutation similar to a matrix $\tilde{Z}=\diag[\tilde{Z}_{1},\cdots,\tilde{Z}_{\lfloor \frac{N+1}{2} \rfloor}]$ where $\tilde{Z}_{1}\in\Lambda,\;\;\Pi,\;\;\text{or}\;\; \Xi$, and $\tilde{Z}_{j}\in \Xi$, $ 2\le j\le \lfloor \frac{N+1}{2} \rfloor$. Let $\mathcal{P}=\mathcal{P}_{2}  \mathcal{P}_{1}$. Then we have $Z= \mathcal{P}_{1} ^{\top} \bar{Z} \mathcal{P}_{1}= \mathcal{P}_{1} ^{\top}\mathcal{P}_{2}^{\top}\tilde{Z}\mathcal{P}_{2}  \mathcal{P}_{1} =\mathcal{P} ^{\top} \tilde{Z} \mathcal{P}$. Obviously, $\mathcal{P}$ is a permutation matrix. This completes the necessity part of the proof.

To establish the sufficiency, suppose that the Gaussian graph matrix $Z$  of a pure Gaussian state is permutation similar to a block diagonal matrix $\tilde{Z}$, i.e., $Z =\mathcal{P}^{\top}\tilde{Z}\mathcal{P}$, 
 where $\mathcal{P}\in\mathbb{R}^{N\times N}$ is a permutation matrix and $\tilde{Z}=\diag[\tilde{Z}_{1},\cdots,\tilde{Z}_{\lfloor \frac{N+1}{2} \rfloor}]$ has been specified in Theorem~\ref{thm1}. We now construct a  linear quantum system that satisfies the two constraints~\ref{constraint1} and~\ref{constraint2},  and that also generates the given pure Gaussian state. First, we construct a  diagonal matrix $\tilde{R}=\diag[\tilde{R}_{1}, \; \cdots, \; \tilde{R}_{\lfloor \frac{N+1}{2} \rfloor}]$, where
 \begin{equation} \label{construction}
 \tilde{R}_{1}= \left\{ 
 \begin{aligned}
 &0, \hspace{1.8cm} \text{if}    \; \tilde{Z}_{1}\in\Lambda,\\
 &\diag[0,\;1],\hspace{0.57cm} \text{if}    \; \tilde{Z}_{1}\in\Pi\; \text{and}\; \tilde{Z}_{1}\ne iI_{2},\\
 &\diag[1,\;-1], \hspace{0.15cm} \;\; \text{if}\;\; \tilde{Z}_{1} \in  \Xi, 
 \end{aligned}\right. 
    \end{equation}
    and $ \tilde{R}_{j}=\diag[j,\;-j], \quad  2\le j\le \lfloor \frac{N+1}{2} \rfloor$.
Then by direct substitution, we can verify that $-\tilde{Z}_{j}\tilde{R}_{j}\tilde{Z}_{j}= \tilde{R}_{j}$, $1\le j\le \lfloor \frac{N+1}{2} \rfloor $. Hence $-\tilde{Z}\tilde{R} \tilde{Z}= \tilde{R}$. In Lemma~\ref{lem1}, we choose $R=\mathcal{P}^{\top}\tilde{R}\mathcal{P}$ and $\Gamma=XRY$. Then $R$ is a diagonal matrix and $-ZRZ=R$. It follows that $YRY-XRX=R$ and $XRY+YRX=0$. Hence we have $\Gamma+\Gamma^{\top}=0$, i.e., $\Gamma$ is a skew symmetric matrix. Substituting $R$ and $\Gamma$ into~\eqref{G} yields $G=\diag[R,\;\;R]$. Therefore  the resulting system Hamiltonian $\hat{H}=\frac{1}{2}\hat{x}^{\top}G\hat{x}$  satisfies the first constraint~\ref{constraint1}. 
 
Next we show that the matrix  $Q$ in~\eqref{rankconstraint} is a non-derogatory matrix. Since $Q=-RZ=-\mathcal{P}^{\top}\tilde{R}\tilde{Z}\mathcal{P}$, using Lemma~\ref{similar}, we have to show that $\tilde{R}\tilde{Z}$ is a non-derogatory matrix.  Note that $\tilde{R}\tilde{Z}=\diag[\tilde{R}_{1} \tilde{Z}_{1}, \; \cdots, \; \tilde{R}_{\lfloor \frac{N+1}{2} \rfloor} \tilde{Z}_{\lfloor \frac{N+1}{2} \rfloor}]$. If $\tilde{Z}_{1}\in\Lambda$, then according to~\eqref{construction}, we have $\tilde{R}_{1} \tilde{Z}_{1}=0$.  If $\tilde{Z}_{1} \in \Pi$ and $\tilde{Z}_{1}\ne iI_{2}$, we have $\tilde{R}_{1} \tilde{Z}_{1}=\diag[0,\;i]$.  If $\tilde{Z}_{1}\in\Xi$, by Lemma~\ref{twodimensional}, $\tilde{R}_{1} \tilde{Z}_{1}$ is diagonalizable and the eigenvalues of $\tilde{R}_{1} \tilde{Z}_{1}$ are $\pm i$. Similarly, it can be shown that $\tilde{R}_{j} \tilde{Z}_{j}$ is diagonalizable and the eigenvalues of $\tilde{R}_{j} \tilde{Z}_{j}$ are $\pm ji$, $2\le j \le \lfloor \frac{N+1}{2} \rfloor $. From the discussions above, we conclude that the block diagonal matrix $\tilde{R}\tilde{Z}$ is diagonalizable and all of its eigenvalues are distinct. Therefore, using Definition~\ref{def}, it can be shown that $\tilde{R}\tilde{Z}$ is a non-derogatory matrix. As a result, $Q$ is also a non-derogatory matrix. Using Lemma~\ref{non-dero}, we can always find a column vector $P\in \mathbb{C}^{N}$ such that the rank condition~\eqref{rankconstraint} is satisfied. Substituting this $P$ into~\eqref{C}, we will obtain a desired coupling vector $\hat{L}$ that satisfies the second constraint~\ref{constraint2}. This completes the proof. 
\end{proof}

Depending on whether $N$ is even or not, we can distinguish two cases for Theorem~\ref{thm1}. The result is summarized in the following corollary. 
\begin{coro} \label{coro0}
Given an $N$-mode pure Gaussian state, (\textrm{i}) if $N$ is even, the pure Gaussian state can be  generated by a  linear quantum system subject to the two constraints~\ref{constraint1} and~\ref{constraint2} if and only if its Gaussian graph matrix $Z$ can be written as  
\begin{align*}
Z=\mathcal{P}^{\top}\tilde{Z}\mathcal{P},\quad \tilde{Z}=\diag[\tilde{Z}_{1},\cdots,\tilde{Z}_{ N/2 }],
\end{align*}
 where $\mathcal{P}\in\mathbb{R}^{N\times N}$ is a permutation matrix, $\tilde{Z}_{1}\in\left(\Pi\cup\Xi\right)$, and $\tilde{Z}_{j}\in \Xi$, $ j=2,\cdots, N/2$; (\textrm{ii}) 
if $N$ is odd, the  pure Gaussian state can be  generated by a  linear quantum system subject to the two constraints~\ref{constraint1} and~\ref{constraint2} if and only if its Gaussian graph matrix $Z$ can be written as 
\begin{align*}
Z=\mathcal{P}^{\top}\tilde{Z}\mathcal{P},\quad \tilde{Z}=\diag[\tilde{Z}_{1},\cdots,\tilde{Z}_{ (N+1)/2 }],
\end{align*}
 where $\mathcal{P}\in\mathbb{R}^{N\times N}$ is a permutation matrix, $\tilde{Z}_{1}\in\Lambda$, and $\tilde{Z}_{j}\in \Xi$, $ j=2,\cdots, (N+1)/2 $.
 \end{coro}

Next, we give an equivalent description for $\Xi$. 
\begin{thm} \label{lemset}
\begin{align*}
\Xi=\Phi,
\end{align*} 
where  
\begin{align*}
&    \Phi \triangleq \Bigg\{\begin{bmatrix}
z_{11} &z_{12}\\
z_{12} &z_{11}
\end{bmatrix}\in\mathbb{C}^{2\times 2}\Bigg|\; z_{12}^{2}=z_{11}^{2}+1 \;\; \text{and} \;\; \im (z_{11})>0\Bigg\}. 
\end{align*}
\end{thm}

\begin{proof}
Suppose $\mathpzc{Z}=\begin{bmatrix}
z_{11} &z_{12}\\
z_{12} &z_{22}
\end{bmatrix}\in \Xi$. Because $\im\left(\mathpzc{Z}\right)>0$, we have $\im\left(z_{11}\right)>0$ and $\im\left(z_{22}\right)>0$. Let us consider the equation $\bigg(\diag[1,-1]\mathpzc{Z} \bigg)^{2}=-I_{2}$  described in $\Xi$. Substituting $\mathpzc{Z}=\begin{bmatrix}
z_{11} &z_{12}\\
z_{12} &z_{22}
\end{bmatrix}$ into the equation  $\bigg(\diag[1,-1]\mathpzc{Z} \bigg)^{2}=-I_{2}$, we obtain
\begin{numcases}{}
z_{11}^{2}-z_{12}^{2}=-1,  \label{chapter3_e21}\\
z_{11}z_{12}-z_{12}z_{22}=0, \label{chapter3_e22}\\
z_{22}^{2} -z_{12}^{2}=-1.  \label{chapter3_e23}
\end{numcases}
If $z_{12}=0$, because $\im(\mathpzc{Z})>0$, we have $z_{11}=z_{22}=i$. If $z_{12}\ne 0$, it follows from~\eqref{chapter3_e22} that $z_{11}=z_{22}$. Therefore, in both cases, we have $\mathpzc{Z}=\begin{bmatrix}
z_{11} &z_{12}\\
z_{12} &z_{11}
\end{bmatrix}$. Next we show that if $\im (z_{11})>0$ and $z_{11}^{2}-z_{12}^{2}=-1$, then the condition $\im(\mathpzc{Z})>0$ is always satisfied. Let us assume $z_{11}=\mu_{1}+i\nu_{1}$, where $\nu_{1}>0$, and $z_{12}=\mu_{2}+i\nu_{2}$. Then it follows from~\eqref{chapter3_e21} that 
\begin{numcases}{}
\mu_{1}^{2}-\nu_{1}^{2}-\mu_{2}^{2}+\nu_{2}^{2}=-1,  \label{chapter3_e24}\\
\mu_{1}\nu_{1}- \mu_{2}\nu_{2}=0.  \label{chapter3_e25}
\end{numcases}
Multiplying both sides of Equation~\eqref{chapter3_e24} by $\nu_{1}^{2}$, we obtain 
\begin{align*}
\mu_{1}^{2}\nu_{1}^{2}-\nu_{1}^{4}-\mu_{2}^{2}\nu_{1}^{2}+\nu_{1}^{2}\nu_{2}^{2}=-\nu_{1}^{2}.
\end{align*}
Using \eqref{chapter3_e25}, we have 
\begin{align*}
\mu_{2}^{2}\nu_{2}^{2}-\nu_{1}^{4}-\mu_{2}^{2}\nu_{1}^{2}+\nu_{1}^{2}\nu_{2}^{2}&=-\nu_{1}^{2},\\
\left(\mu_{2}^{2}+\nu_{1}^{2}\right)\left(\nu_{2}^{2}-\nu_{1}^{2}\right)&=-\nu_{1}^{2},\\
 \nu_{1}^{2}-\nu_{2}^{2}&=\frac{\nu_{1}^{2}}{\mu_{2}^{2}+\nu_{1}^{2}}.
\end{align*}
Since $\nu_{1}>0$, we have $\nu_{1}^{2}-\nu_{2}^{2}>0$. That is, $\im(\mathpzc{Z})>0$.  Therefore, we have  
\begin{align*}
\Xi = \Bigg\{\begin{bmatrix}
z_{11} &z_{12}\\
z_{12} &z_{11}
\end{bmatrix}\Bigg|\; z_{12}^{2}=z_{11}^{2}+1 \;\; \text{and} \;\; \im (z_{11})>0\Bigg\}. 
\end{align*}
That is, $\Xi = \Phi$. This completes the proof.  
\end{proof} 

The following result is an immediate application  of Theorem~\ref{lemset}. It gives a complete parametrization of the class of pure Gaussian states that can be generated by  linear quantum systems subject to the  constraints~\ref{constraint1} and~\ref{constraint2}.
\begin{coro} \label{coro2}
An $N$-mode pure Gaussian state can be  generated by a  linear quantum system subject to the two constraints~\ref{constraint1} and~\ref{constraint2} if and only if its Gaussian graph matrix $Z$ can be written as 
\begin{align*}
Z=\mathcal{P}^{\top}\tilde{Z}\mathcal{P},\quad \tilde{Z}=\diag[\tilde{Z}_{1},\cdots,\tilde{Z}_{\lfloor \frac{N+1}{2} \rfloor}],
\end{align*}
where $\mathcal{P}\in\mathbb{R}^{N\times N}$ is a permutation matrix, \small{$\tilde{Z}_{1}\in\left(\Lambda\cup\Pi\cup\Phi \right)$}, and $\tilde{Z}_{j}\in \Phi $, $ j=2,\cdots,  \lfloor\frac{N+1}{2} \rfloor$.
\end{coro}

\normalsize

The following corollary gives an equivalent statement of Corollary~\ref{coro2}. 
\begin{coro} \label{coro3}
 Given an $N$-mode pure Gaussian state, (\textrm{i}) if $N$ is even, the pure Gaussian state can be  generated by a  linear quantum system subject to the two constraints~\ref{constraint1} and~\ref{constraint2} if and only if its Gaussian graph matrix $Z$ can be written as   
\begin{align*}
Z=\mathcal{P}^{\top}\tilde{Z}\mathcal{P},\quad \tilde{Z}=\diag[\tilde{Z}_{1},\cdots,\tilde{Z}_{ N/2 }],
\end{align*}
 where $\mathcal{P}\in\mathbb{R}^{N\times N}$ is a permutation matrix, $\tilde{Z}_{1}\in\left(\Pi\cup\Phi \right)$, and $\tilde{Z}_{j}\in \Phi $, $j=2,\cdots, N/2$; (\textrm{ii}) 
if $N$ is odd, the  pure Gaussian state can be  generated by a  linear quantum system subject to the two constraints~\ref{constraint1} and~\ref{constraint2} if and only if its Gaussian graph matrix $Z$ can be written as 
\begin{align*}
Z=\mathcal{P}^{\top}\tilde{Z}\mathcal{P},\quad \tilde{Z}=\diag[\tilde{Z}_{1},\cdots,\tilde{Z}_{ (N+1)/2 }],
\end{align*}
 where $\mathcal{P}\in\mathbb{R}^{N\times N}$ is a permutation matrix, $\tilde{Z}_{1}\in\Lambda$, and $\tilde{Z}_{j}\in \Phi $, $j=2,\cdots, (N+1)/2$.
\end{coro}

\section{Example}\label{Example}
\begin{exam} \label{Example2}
To illustrate the main result of the paper, we consider the generation of {\it  two-mode squeezed states}~\cite{WC14:pra,MFL11:pra}. Two-mode squeezed states are  an important resource in several quantum information protocols such as quantum cryptography and  quantum 
teleportation.  The covariance matrix of a canonical two-mode squeezed state is given by \begin{align*} 
V=\frac{1}{2}\begin{bmatrix}
\cosh(2\alpha) &\sinh(2\alpha) &0 &0\\
\sinh(2\alpha) &\cosh(2\alpha) &0 &0\\
0 &0 &\cosh(2\alpha) &-\sinh(2\alpha)\\
0 &0 &-\sinh(2\alpha) &\cosh(2\alpha)
\end{bmatrix},
\end{align*}
where $\alpha $ is the squeezing parameter.  Applying the factorization~\eqref{covariance}, we obtain the Gaussian graph matrix for a canonical two-mode squeezed state, i.e.,  $
Z=X+iY$, where $X=0$ and $Y=\begin{bmatrix}
\cosh(2\alpha)  &-\sinh(2\alpha)\\
-\sinh(2\alpha) &\cosh(2\alpha)  
\end{bmatrix}$. By direct calculation, we find $Z\in\Xi$. According to Theorem~\ref{thm1}, every two-mode squeezed state can be  generated by a linear  quantum system subject to the two constraints~\ref{constraint1} and~\ref{constraint2}. To construct such a system, let us choose $R=\diag[1,\; -1]$, $\Gamma=XRY=0$ and $P=i\frac{\cosh(\alpha)+\sinh(\alpha)}{\sqrt{2}}\begin{bmatrix}
1 &1\end{bmatrix}^{\top}$. It can be verified by direct substitution that the rank condition~\eqref{rankconstraint} holds. Then based on Lemma~\ref{lem1}, the resulting  linear quantum system is strictly stable and  generates the desired two-mode squeezed state.  
The system Hamiltonian $\hat{H}$ is given by
$\hat{H}=\frac{1}{2}\hat{x}^{\top}G\hat{x}=\frac{1}{2}(\hat{q}_{1}^{2}+\hat{p}_{1}^{2}-\hat{q}_{2}^{2}-\hat{p}_{2}^{2})=\hat{a}_{1}^{\ast}\hat{a}_{1}-\hat{a}_{2}^{\ast}\hat{a}_{2}$, where $\hat{a}_{j}=(\hat{q}_{j}+i\hat{p}_{j})/\sqrt{2}$ and $\hat{a}_{j}^{\ast}=(\hat{q}_{j}-i\hat{p}_{j})/\sqrt{2}$ denote the annihilation operator and creation operator for the $j$th mode, respectively. It can be seen that the system Hamiltonian $\hat{H}$ satisfies the first constraint~\ref{constraint1}. The coupling operator 
 $\hat{L}$ is given by $\hat{L}=C\hat{x}=
\frac{\cosh(\alpha)-\sinh(\alpha)}{\sqrt{2}}\left(\hat{q}_{1}+\hat{q}_{2}\right)+i\frac{\cosh(\alpha)+\sinh(\alpha)}{\sqrt{2}}\left(\hat{p}_{1}+\hat{p}_{2}\right)=\hat{a}_{1}\cosh(\alpha)-\hat{a}_{2}^{\ast}\sinh(\alpha)+\hat{a}_{2}\cosh(\alpha)-\hat{a}_{1}^{\ast}\sinh(\alpha) $, which 
satisfies the constraint~\ref{constraint2}. The coupling operator $\hat{L}$ is one of the nullifiers for the two-mode squeezed state~\cite{MFL11:pra}. Another nullifier is  
$\hat{L}_{2}= \frac{\cosh(\alpha)+\sinh(\alpha)}{\sqrt{2}} \left(\hat{q}_{1}-\hat{q}_{2}\right)+i\frac{\cosh(\alpha)-\sinh(\alpha)}{\sqrt{2}} \left(\hat{p}_{1}-\hat{p}_{2}\right)\; =\hat{a}_{1}\cosh(\alpha)-\hat{a}_{2}^{\ast}\sinh(\alpha)-\hat{a}_{2}\cosh(\alpha)+\hat{a}_{1}^{\ast}\sinh(\alpha) $. We observe that $\hat{H}=\frac{1}{2}(\hat{L}^{\ast}\hat{L}_{2}+\hat{L}_{2}^{\ast}\hat{L})$ is a beam-splitter-like interaction between the two nullifiers.  The amount of entanglement contained in a two-mode squeezed state can be quantified  using the logarithmic negativity $\mathcal{E}$~\cite{VW02:pra}.  This value is found to be $\mathcal{E}=2|\alpha |$. 

We mention that the above proposal is an idealization, since in any physical realization of a linear quantum system the existence of additional thermal noises is unavoidable. To describe this case, we include some auxiliary thermal noise inputs via additional coupling operators. That is, $
\hat{L}_{1\uparrow}=\sqrt{\gamma_{1}\bar{n}_{1}} \hat{a}_{1}^{\ast}$, $
\hat{L}_{1\downarrow}=\sqrt{\gamma_{1}\left(\bar{n}_{1}+1\right)}\hat{a}_{1}$, $\hat{L}_{2\uparrow}=\sqrt{\gamma_{2}\bar{n}_{2}}\hat{a}_{2}^{\ast}$, 
$\hat{L}_{2\downarrow}=\sqrt{\gamma_{2}\left(\bar{n}_{2}+1\right)}\hat{a}_{2}$, where  $\gamma_{1}$, $\gamma_{2}$ are damping rates and  $\bar{n}_{1}$, $\bar{n}_{2}$ denote the thermal occupations of the environments. This is indeed the case for an optomechanical realization~\cite{WC14:pra,ODPCWS16:prl,NTMPS17:pnas},  as described in Fig.~\ref{fig1}.  Due to these thermal noise inputs, the steady state generated is not exactly a two-mode squeezed state. For example, if we take $\alpha=0.7$, $\bar{n}_{1}=\bar{n}_{2}=10$ and $\gamma_{1}=\gamma_{2}=0.01$, the two-mode steady state generated by the linear quantum system with Hamiltonian $\hat{H}$ and coupling operators $\hat{L}$, $\hat{L}_{1\uparrow}$, $\hat{L}_{1\downarrow}$, $\hat{L}_{2\uparrow}$, $\hat{L}_{2\downarrow}$  has the covariance matrix
\begin{align*}
V=\begin{bmatrix}
    1.1943   & 0.9169  & -0.0047  & -0.0464\\
    0.9169   & 1.1943   & 0.0464   & 0.0047\\
   -0.0047   & 0.0464   & 1.1896  & -0.9638\\
   -0.0464   & 0.0047  & -0.9638   & 1.1896
\end{bmatrix}.
\end{align*} The purity of this state is found to be $p=0.4787$ and the entanglement value, quantified via the logarithmic negativity, is $\mathcal{E}=0.7134$. For comparison, the two-mode steady state generated by the linear quantum system with  Hamiltonian $\hat{H}$ and only the thermal noises $\hat{L}_{1\uparrow}$, $\hat{L}_{1\downarrow}$, $\hat{L}_{2\uparrow}$, $\hat{L}_{2\downarrow}$  has the covariance matrix
\begin{align} \label{VVV}
V=\begin{bmatrix}
    10.5   &0  &0  &0\\
    0   &10.5  &0   &0\\
   0   &0   &10.5  &0\\
   0   &0  &0  &10.5
\end{bmatrix}.
\end{align} The purity of this state, without the designed coupling operator $\hat{L}$, is found to be $p=0.0023$ and the entanglement is  $\mathcal{E}=0$. We see that  our designed coupling operator $\hat{L}$ both generates the entanglement and increases the purity of the steady state; see~\cite{WC14:pra} for more discussion on this issue.  

\begin{figure}[htbp]
\begin{center}
\includegraphics[height=2.5cm]{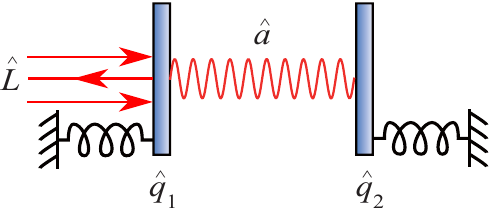}
\caption{
The optomechanical linear quantum  system that  generates the canonical two-mode squeezed state.  Two mechanical oscillators (denoted by $\hat{q}_{1}$ and $\hat{q}_{2}$, respectively) are
 coupled to a single cavity mode (denoted by $\hat{a}$). The cavity mode is used to  construct the required coupling operator $\hat{L}$. It responds rapidly to the mechanical motion and can be adiabatically eliminated to give an effective description of the dynamics of the two mechanical oscillator modes alone.}
\label{fig1}
\end{center}
\end{figure}
\end{exam}

\begin{exam}
Consider a two-mode pure Gaussian state with covariance matrix given by 
\begin{align*}
V=\begin{bmatrix}
\frac{\sqrt{6}}{2}  &-1 &0 &\frac{1}{2}\\
-1 &\frac{\sqrt{6}}{2}  &\frac{1}{2} &0\\
0 &\frac{1}{2} &\frac{\sqrt{6}}{2}  &1\\
\frac{1}{2} &0 &1 &\frac{\sqrt{6}}{2}
\end{bmatrix}.
\end{align*} Note that this state has entanglement between the two modes. Applying the factorization~\eqref{covariance}, we obtain the Gaussian graph matrix for this pure Gaussian state, i.e.,  $
Z=X+iY$, where $X=\begin{bmatrix}
1 &\frac{\sqrt{6}}{2}\\
\frac{\sqrt{6}}{2} &1
\end{bmatrix}$ and $Y=\begin{bmatrix}
\frac{\sqrt{6}}{2}  &1\\
1 &\frac{\sqrt{6}}{2}  
\end{bmatrix}$. By direct calculation, we find $Z\in\Xi$. According to Theorem~\ref{thm1}, the pure Gaussian state can be  generated by a linear  quantum system subject to the two constraints~\ref{constraint1} and~\ref{constraint2}. To construct such a system, we choose $R=\diag[1,\; -1]$, $\Gamma=XRY=\begin{bmatrix}
0  &-\frac{1}{2}\\
\frac{1}{2} &0
\end{bmatrix}$ and $P=[0\; \;\;1]^{\top}$. It can be verified by direct substitution that the rank condition~\eqref{rankconstraint} holds. Then by Lemma~\ref{lem1}, the resulting  linear quantum system is strictly stable and  generates the desired pure Gaussian state.  
The system Hamiltonian $\hat{H}$ is given by
$\hat{H}=\frac{1}{2}\hat{x}^{\top}G\hat{x}=\frac{1}{2}(\hat{q}_{1}^{2}+\hat{p}_{1}^{2}-\hat{q}_{2}^{2}-\hat{p}_{2}^{2})$, which satisfies the constraint~\ref{constraint1}. 
The coupling operator $\hat{L}$ is  $
\hat{L}=C\hat{x}=-\left(\frac{\sqrt{6}}{2}+i\right) \hat{q}_{1}-\left(1+\frac{\sqrt{6}}{2}i\right)\hat{q}_{2} +\hat{p}_{2}$,
which satisfies the constraint~\ref{constraint2}. An optomechanical realization of this system is similar to the realization shown in Fig.~\ref{fig1}, and hence is omitted.  The amount of entanglement contained in the state 
 is found to be $\mathcal{E}=1.5445$. Thus, the two modes are highly entangled. 

Similarly, we can add some auxiliary thermal noise inputs to the linear quantum system above to provide a more realistic model of a physical system, i.e., we add the coupling operators $
\hat{L}_{1\uparrow}=\sqrt{\gamma_{1}\bar{n}_{1}} \hat{a}_{1}^{\ast}$, $
\hat{L}_{1\downarrow}=\sqrt{\gamma_{1}\left(\bar{n}_{1}+1\right)}\hat{a}_{1}$, $\hat{L}_{2\uparrow}=\sqrt{\gamma_{2}\bar{n}_{2}}\hat{a}_{2}^{\ast}$, 
$\hat{L}_{2\downarrow}=\sqrt{\gamma_{2}\left(\bar{n}_{2}+1\right)}\hat{a}_{2}$. Let us take $\bar{n}_{1}=\bar{n}_{2}=10$ and $\gamma_{1}=\gamma_{2}=0.01$. In this case, the two-mode steady state generated by the linear quantum system with Hamiltonian $\hat{H}$ and coupling operators $\hat{L}$, $\hat{L}_{1\uparrow}$, $\hat{L}_{1\downarrow}$, $\hat{L}_{2\uparrow}$, $\hat{L}_{2\downarrow}$  has the covariance matrix
\begin{align*}
V=\begin{bmatrix}
    1.7298   &-1.2921   &-0.0439   & 0.4866\\
   -1.2921   & 1.4780   & 0.6209   & 0.0270\\
   -0.0439   & 0.6209   & 1.5698   & 1.0281\\
    0.4866   & 0.0270   & 1.0281   & 1.2790
\end{bmatrix}.
\end{align*} The purity of this state is $p=0.4175$ and the entanglement value is $\mathcal{E}=0.8479$. For comparison, the two-mode steady state generated by the linear quantum system with Hamiltonian $\hat{H}$ and only the thermal noises $\hat{L}_{1\uparrow}$, $\hat{L}_{1\downarrow}$, $\hat{L}_{2\uparrow}$, $\hat{L}_{2\downarrow}$  has the covariance matrix $V$ as given in~\eqref{VVV}. The purity of this state is $p=0.0023$ and the entanglement value, quantified via the logarithmic negativity, is $\mathcal{E}=0$. Hence  the two-mode steady state generated with the designed coupling operator $\hat{L}$ is highly pure and entangled compared to the case without  $\hat{L}$.  

\end{exam}
\begin{exam}
We consider an eight-mode pure Gaussian state with covariance matrix given by $
V=\left[ {\begin{array}{*{20}c} V_{11} &\vline &V_{12}  \\
\hline
V_{12}^{\top} &\vline &V_{22} 
\end{array} } \right]$, 
where 
\begin{align*}
&V_{11}=\\
&\diag\left[\begin{bmatrix}
\frac{\sqrt{6}}{2}  &-1 \\
-1 &\frac{\sqrt{6}}{2} 
\end{bmatrix}, \begin{bmatrix}
\frac{\sqrt{6}}{2}  &-1 \\
-1 &\frac{\sqrt{6}}{2} 
\end{bmatrix},\begin{bmatrix}
\frac{\sqrt{6}}{2}  &-1 \\
-1 &\frac{\sqrt{6}}{2} 
\end{bmatrix},\begin{bmatrix}
\frac{\sqrt{6}}{2}  &-1 \\
-1 &\frac{\sqrt{6}}{2} 
\end{bmatrix}\right],\\
&V_{22}=\\
&\diag\left[\begin{bmatrix}
\frac{\sqrt{6}}{2}  &1 \\
1 &\frac{\sqrt{6}}{2} 
\end{bmatrix}, \begin{bmatrix}
\frac{\sqrt{6}}{2}  &1 \\
1 &\frac{\sqrt{6}}{2} 
\end{bmatrix},\begin{bmatrix}
\frac{\sqrt{6}}{2}  &1 \\
1 &\frac{\sqrt{6}}{2} 
\end{bmatrix},\begin{bmatrix}
\frac{\sqrt{6}}{2}  &1 \\
1 &\frac{\sqrt{6}}{2} 
\end{bmatrix}\right],\\
&\text{and}\;\;V_{12}=\diag\left[\begin{bmatrix}
0  &\frac{1}{2} \\
\frac{1}{2} &0 
\end{bmatrix},\begin{bmatrix}
0  &\frac{1}{2} \\
\frac{1}{2} &0 
\end{bmatrix},\begin{bmatrix}
0  &\frac{1}{2} \\
\frac{1}{2} &0 
\end{bmatrix},\begin{bmatrix}
0  &\frac{1}{2} \\
\frac{1}{2} &0 
\end{bmatrix}\right].
\end{align*}
 Applying the factorization~\eqref{covariance}, we obtain the Gaussian graph matrix for this pure Gaussian state, i.e.,  $
Z=X+iY$, where 
\begin{align*}
&X=\\
&\diag\left[\begin{bmatrix}
1 &\frac{\sqrt{6}}{2}\\
\frac{\sqrt{6}}{2}   &1
\end{bmatrix}, \begin{bmatrix}
1 &\frac{\sqrt{6}}{2}\\
\frac{\sqrt{6}}{2}   &1
\end{bmatrix},\begin{bmatrix}
1 &\frac{\sqrt{6}}{2}\\
\frac{\sqrt{6}}{2}   &1
\end{bmatrix},\begin{bmatrix}
1 &\frac{\sqrt{6}}{2}\\
\frac{\sqrt{6}}{2}   &1
\end{bmatrix}\right],\\
&Y=\\
&\diag\left[\begin{bmatrix}
\frac{\sqrt{6}}{2}  &1 \\
1 &\frac{\sqrt{6}}{2} 
\end{bmatrix}, \begin{bmatrix}
\frac{\sqrt{6}}{2}  &1 \\
1 &\frac{\sqrt{6}}{2} 
\end{bmatrix},\begin{bmatrix}
\frac{\sqrt{6}}{2}  &1 \\
1 &\frac{\sqrt{6}}{2} 
\end{bmatrix},\begin{bmatrix}
\frac{\sqrt{6}}{2}  &1 \\
1 &\frac{\sqrt{6}}{2} 
\end{bmatrix}\right].
\end{align*}
By Theorem~\ref{thm1}, the pure Gaussian state can be  generated by a linear  quantum system subject to the two constraints~\ref{constraint1} and~\ref{constraint2}. To construct such a system, we choose 
\begin{align*}
R&=\diag[1,\; -1,\;\; 2,\; -2,\;\;3,\; -3,\;\;4,\; -4],\\
\Gamma&=
\diag\left[\begin{bmatrix}
0  &-\frac{1}{2}  \\
\frac{1}{2} &0
\end{bmatrix}, \begin{bmatrix}
0  &-1 \\
1 &0
\end{bmatrix},\begin{bmatrix}
0  &-\frac{3}{2} \\
\frac{3}{2} &0 
\end{bmatrix},\begin{bmatrix}
0  &-2 \\
2 &0
\end{bmatrix}\right],\\
P&=\begin{bmatrix}0  & 1 & 0 &1 &0 &1 &0 &1\end{bmatrix}^{\top}.
\end{align*} It can be verified by direct substitution that the rank condition~\eqref{rankconstraint} holds. According to Lemma~\ref{lem1}, the resulting  linear quantum system is strictly stable and  generates the desired pure Gaussian state.  
The system Hamiltonian $\hat{H}$ is given by
$\hat{H}=\frac{1}{2}\hat{x}^{\top}G\hat{x}=\frac{1}{2}(\hat{q}_{1}^{2}+\hat{p}_{1}^{2}-\hat{q}_{2}^{2}-\hat{p}_{2}^{2})+(\hat{q}_{3}^{2}+\hat{p}_{3}^{2}-\hat{q}_{4}^{2}-\hat{p}_{4}^{2})+\frac{3}{2}(\hat{q}_{5}^{2}+\hat{p}_{5}^{2}-\hat{q}_{6}^{2}-\hat{p}_{6}^{2})+2(\hat{q}_{7}^{2}+\hat{p}_{7}^{2}-\hat{q}_{8}^{2}-\hat{p}_{8}^{2})$, which satisfies the constraint~\ref{constraint1}. 
The coupling operator $\hat{L}$ is given by 
\begin{align*}
\hat{L}&=-\left(\frac{\sqrt{6}}{2}+i\right)( \hat{q}_{1}+ \hat{q}_{3}+ \hat{q}_{5}+ \hat{q}_{7})\\
&-\left(1+\frac{\sqrt{6}}{2}i\right)(\hat{q}_{2}+\hat{q}_{4}+\hat{q}_{6}+\hat{q}_{8} ) +\hat{p}_{2}+\hat{p}_{4} +\hat{p}_{6}+\hat{p}_{8},
\end{align*}
which satisfies the constraint~\ref{constraint2}. Finally, let us quantify the amount of entanglement contained in the given  state using the logarithmic negativity $\mathcal{E}$. It is found that the eight modes of the system can be divided into four groups, i.e., $(1,2)$, $(3,4)$, $(5,6)$ and $(7,8)$. When the system reaches its steady state, the two modes in each group are highly entangled (i.e., $\mathcal{E}_{(2j+1,2j+2)}=1.5445$, $j=0,1,2,3$), but no entanglement exists between the different groups. 
\end{exam}

\begin{rmk}
For an $N$-mode pure Gaussian state ($N\ge 2$), it is straightforward to verify that if it can be  generated by a linear  quantum system subject to the two constraints~\ref{constraint1} and~\ref{constraint2}, then we can divide the $N$ modes into $\lfloor \frac{N+1}{2} \rfloor$ groups with each group consisting of no more than $2$ modes.  Pairwise entanglement may exist between the two modes in the same group. However, no entanglement exists between the different groups. From Theorem~\ref{thm1}, it can be seen that the $N$-mode pure Gaussian state is the tensor product of several one or two mode pure Gaussian states, each of which can be  generated by a linear  quantum system subject to the two constraints~\ref{constraint1} and~\ref{constraint2}.  Therefore, the quantum system generating this $N$-mode pure Gaussian state can be regarded as a combination of the $\lfloor \frac{N+1}{2} \rfloor$ linear  quantum subsystems subject to the two constraints~\ref{constraint1} and~\ref{constraint2}. 
\end{rmk}

\section{Conclusion} \label{Conclusion}
In this paper, we consider linear quantum systems subject to constraints. First, we assume that the system Hamiltonian $\hat{H}$ is of the form $\hat{H}=\sum_{j=1}^{N}  \frac{\omega_{j}}{2}\left(\hat{q}_{j}^{2}+\hat{p}_{j}^{2}\right)$, $\omega_{j}\in \mathbb{R}$, $j=1,\cdots, N$. Second, we assume that the system is only coupled to a single dissipative environment. We then parametrize the class of pure Gaussian states that can be  generated by this particular type of linear quantum system. 

It should be mentioned that in any physical realization of a linear quantum system,  the existence of thermal  noises cannot be avoided. In future work, it would be interesting to investigate  the impact of including auxiliary  thermal noises on the steady-state  purity and entanglement. It would also be interesting to address the problem of how to design a coupling operator $\hat{L}$ such that the resulting linear quantum system is optimally robust against those thermal noise inputs. 

Another possible extension is to study the generation of pure Gaussian states under different system constraints. For example,  we have  recently considered a chain consisting of $(2\aleph+1)$ quantum harmonic oscillators with passive nearest-neighbour Hamiltonian interactions and with a single reservoir which acts locally on the central oscillator~\cite{MWPY17:jpa}. We have derived
a necessary  and sufficient condition for a pure Gaussian state to be generated by this type of quantum harmonic oscillator chain.
%%%%%%%%%%%%%%%
%%%%%%%%%%%%%%%%%%%
%%%%%%%%%%%%%%%%%%%%%%%%%
\bibliographystyle{ieeetr} 

\end{document}